\documentclass[preprint,12pt]{elsarticle}

\usepackage{amssymb}
\usepackage{amsthm}
\usepackage{balance}
\newtheorem{thm}{Theorem}[section]

\newtheorem{lem}[thm]{Lemma}
\newtheorem{defn}[thm]{Definition}
\newtheorem{remark}[thm]{Remark}
\journal{arXiv}

\begin{document}

\begin{frontmatter}

\title{On the Windfall and Price of Friendship:\\Inoculation Strategies on Social Networks}

\author{Dominic Meier$^1$, Yvonne Anne Pignolet$^2$,\\Stefan Schmid$^3$, Roger Wattenhofer$^1$\\
$^1$ ETH Zurich, Switzerland\\
$^2$ ABB Research, Switzerland\\
$^3$ T-Labs / TU Berlin, Germany\\
}

\address{}

\begin{abstract}
This article investigates selfish behavior in games where players
are embedded in a social context. A framework is presented which
allows us to measure the \emph{Windfall of Friendship}, i.e., how
much players benefit (compared to purely selfish environments) if
they care about the welfare of their friends in the social network
graph. As a case study, a virus inoculation game is examined. We
analyze the corresponding Nash equilibria and show that the Windfall
of Friendship can never be negative. However, we find that if the
valuation of a friend is independent of the total number of friends,
the social welfare may not increase monotonically with the extent to
which players care for each other; intriguingly, in the
corresponding scenario where the relative importance of a friend
declines, the Windfall is monotonic again. This article also studies
convergence of best-response sequences. It turns out that in social
networks, convergence times are typically higher and hence
constitute a price of friendship. While such phenomena may be known
on an anecdotal level, our framework allows us to quantify these
effects analytically. Our formal insights on the worst case
equilibria are complemented by simulations shedding light onto the
structure of other equilibria.
\end{abstract}

\begin{keyword}
Game Theory, Social Networks, Equilibria, Virus Propagation,
Windfall of Friendship
\end{keyword}

\end{frontmatter}

\section{Introduction}\label{sec:intro}

Social networks have existed for thousands of years, but it was not
until recently that researchers have started to gain scientific
insights into phenomena like the \emph{small world property}. The
rise of the Internet has enabled people to connect with each other
in new ways and to find friends sharing the same interests from all
over the planet. A social network on the Internet can manifest
itself in various forms. For instance, on \emph{Facebook}, people
maintain virtual references to their friends. The contacts stored on
mobile phones or email clients form a social network as well. The
analysis of such networks---both their static properties as well as
their evolution over time---is an interesting endeavor, as it
reveals many aspects of our society in general.

A classic tool to model human behavior is \emph{game theory}. It has
been a fruitful research field in economics and sociology for many
years. Recently, computer scientists have started to use game theory
methods to shed light onto the complexities of today's highly
decentralized networks. Game theoretic models traditionally assume
that people act autonomously and are steered by the desire to
maximize their benefits (or utility). Under this assumption, it is
possible to quantify the performance loss of a distributed system
compared to situations where all participants collaborate perfectly.
A widely studied measure which captures this loss of social welfare
is the \emph{Price of Anarchy} (PoA). Even though these concepts can
lead to important insights in many environments, we believe that in
some situations, the underlying assumptions do not reflect reality
well enough. One such example are social networks: most likely
people act less selfishly towards their friends than towards
complete strangers. Such altruistic behavior is typically not
considered in game-theoretic models.

In this article, we propose a game theoretic framework for social
networks. Social networks are not only attractive to their
participants, e.g., it is well-known that the user profiles are an
interesting data source for the PR industry to provide tailored
advertisements. Moreover, social network graphs can also be
exploited for attacks, e.g., email viruses using the users' address
books for propagating, worms spreading on mobile phone networks and
over the Internet telephony tool Skype have been reported (e.g.,
\cite{MobilePhoneMalware}). This article investigates rational
inoculation strategies against such viruses from our game theoretic
perspective, and studies the propagation of such viruses on the
social network.

\subsection{Our Contribution}

This article makes a first step to combine two active threads of
research: social networks and game theory. We introduce a framework
taking into consideration that people may care about the well-being
of their friends. In particular, we define the \emph{Windfall of
Friendship} (WoF) which captures to what extent the social welfare
improves in social networks compared to purely selfish systems.

In order to demonstrate our framework, as a case study, we provide a
game-theoretic analysis of a \emph{virus inoculation game}.
Concretely, we assume that the players have the choice between
inoculating by buying anti-virus software and risking infection. As
expected, our analysis reveals that the players in this game always
benefit from caring about the other participants in the social
network rather than being selfish. Intriguingly, however, we find
that the Windfall of Friendship may not increase monotonically with
stronger relationships. Despite the phenomenon being an
``ever-green'' in political debates, to the best of our knowledge,
this is the first article to quantify this effect formally.

This article derives upper and lower bounds on the Windfall of
Friendship in simple graphs. For example, we show that the Windfall
of Friendship in a complete graph is at most $4/3$; this is tight in
the sense that there are problem instances where the situation can
indeed improve this much. Moreover, we show that in star graphs,
friendship can help to eliminate undesirable equilibria. Generally,
we discover that even in simple graphs the Windfall of Friendship
can attain a large spectrum of values, from constant ratios up to
$\Theta(n)$, $n$ being the network size, which is asymptotically
maximal for general graphs.

Also an alternative friendship model is discussed in this article
where the relative importance of an individual friend declines with
a larger number of friends. While the Windfall of Friendship is
still positive, we show that the non-monotonicity result is no
longer applicable. Moreover, it is proved that in both models,
computing the best and the worst friendship Nash equilibrium is
$\mathcal{NP}$-hard.

The paper also initiates the discussion of implications on
convergence. We give a potential function argument to show
convergence of best-response sequences in various models and for
simple, cyclic graphs. Moreover, we report on our simulations which
indicate that the convergence times are typically higher in social
contexts, and hence constitute a certain price of friendship.

Finally, to complement our formal analysis of the worst equilibria,
simulation results for average case equilibria are discussed.

\subsection{Organization}

The remainder of this article is organized as follows.
Section~\ref{sec:relwork} reviews related work and
Section~\ref{sec:model} formally introduces our model and framework.
The Windfall of Friendship on general graphs and on special graphs
is studied in Sections~\ref{sec:general} and~\ref{sec:cliquestar}
respectively. Section~\ref{sec:relative} discusses an alternative
model where the relative importance of a friend declines if the
total number of friends increases. Aspects of best-response
convergence and implications are considered in
Section~\ref{sec:convergence}. We report on simulations in
Section~\ref{sec:simulations}. Finally, we conclude the article in
Section~\ref{sec:conclusion}.

\section{Related Work}\label{sec:relwork}

Social networks are a fascinating topic not only in social sciences,
but also in ethnology, and psychology. The advent of social networks
on the Internet, e.g., \emph{Facebook}, \emph{LinkedIn},
\emph{MySpace}, \emph{Orkut}, or \emph{Xing}, to name but a few,
heralded a new kind of social interactions, and the mere scale of
online networks and the vast amount of data constitute an
unprecedented treasure for scientific studies. The topological
structure of these networks and the dynamics of the user behavior
has a mathematical and algorithmic dimension, and has raised the
interest of mathematicians and engineers accordingly.

The famous \emph{small world experiment}~\cite{milgram} conducted by
Stanley Milgram 1967 has gained attention by the algorithm
community~\cite{kleinberg} and inspired research on topics such as
decentralized search algorithms~\cite{Kleinberg06, manku04}, routing
on social networks~\cite{Frai04, kleinberg,lino05} and the
identification of communities~\cite{Flake02,newman04}. The dynamics
of epidemic propagation of information or diseases has been studied
from an algorithmic perspective as well~\cite{kleinbergchap,
leskovec2007}. Knowledge on effects of this cascading behavior is
useful to understand phenomena as diverse as word-of-mouth effects,
the diffusion of innovation, the emergence of bubbles in a financial
market, or the rise of a political candidate. It can also help to
identify sets of influential players in networks where marketing is
particularly efficient (\emph{viral marketing}). For a good overview
on economic aspects of social networks, we refer the reader to
\cite{soceco}, which, i.a., compares random graph theory with game
theoretic models for the formation of social networks.

Recently, game theory has also received much attention by computer
scientists. This is partly due to the various actors and
stake-holders who influence the decentralized growth of the
Internet: game theory is a useful tool to gain insights into the
Internet's socio-economic complexity. Many aspects have been studied
from a game-theoretic point of view, e.g., \emph{routing}
\cite{roughgarden,howbad}, \emph{multicast transmissions}
\cite{papamulticast}, or \emph{network creation}
\cite{fabrikant,lg}. Moreover, computer scientists are interested in
the algorithmic problems offered by game theory, e.g., on the
existence of pure equilibria~\cite{papadi}.

This article applies game theory to social networks where players
are not completely selfish and autonomous but have friends about
whose well-being they care to some extent. We demonstrate our
mathematical framework with a virus inoculation game on social
graphs. There is a large body of literature on the propagation of
viruses~\cite{Virus1,Virus2,OtherGraphs2,OtherGraphs1,OtherGraphs3}.
Miscellaneous misuse of social networks has been reported, e.g.,
\emph{email viruses}\footnote{E.g., the Outlook worm
\emph{Worm.ExploreZip}.} have used address lists to propagate to the
users' friends. Similar vulnerabilities have been exploited to
spread worms on the \emph{mobile phone
network}~\cite{MobilePhoneMalware} and on the Internet telephony
tool \emph{Skype}\footnote{See
http://news.softpedia.com/news/Skype-Attacked-By-Fast-Spreading-Virus-52039.shtml.}.

There already exists interesting work on game theoretic and epidemic
models of propagation in social networks. For instance, Montanari
and Saberi~\cite{guest1} attend to a game theoretic model for the
diffusion of an innovation in a network and characterize the rate of
convergence as a function of graph structure. The authors highlight
crucial differences between game theoretic and epidemic models and
find that the spread of viruses, new technologies, and new political
or social beliefs do not have the same viral behavior.

The articles closest to ours are~\cite{virusgame,bg}. Our model is
inspired by Aspnes et al.~\cite{virusgame}. The authors apply a
classic game-theoretic analysis and show that selfish systems can be
very inefficient, as the Price of Anarchy is $\Theta(n)$, where $n$
is the total number of players. They show that computing the social
optimum is $\mathcal{NP}$-hard and give a reduction to the
combinatorial problem \emph{sum-of-squares partition}. They also
present a $O(\log^2{n})$ approximation. Moscibroda et al.~\cite{bg}
have extended this model by introducing malicious players in the
selfish network. This extension facilitates the estimation of the
robustness of a distributed system to malicious attacks. They also
find that in a non-oblivious model, intriguingly, the presence of
malicious players may actually \emph{improve} the social welfare. In
a follow-up work~\cite{sirocco10} which generalizes the social
context of~\cite{bg} to arbitrary bilateral relationships, it has
been shown that there is no such phenomenon in a simple network
creation game. The \emph{Windfall of Malice} has also been studied
in the context of congestion games~\cite{windfallmalice} by Babaioff
et al. In contrast to these papers, our focus here is on social
graphs where players are concerned about their friends' benefits.

There is other literature on game theory where players are
influenced by their neighbors. In \emph{graphical economics}
\cite{geco1,geco2}, an undirected graph is given where an edge
between two players denotes that free trade is allowed between the
two parties, where the absence of such an edge denotes an embargo or
an other restricted form of direct trade. The payoff of a player is
a function of the actions of the players in its neighborhood only.
In contrast to our work, a different equilibrium concept is used and
no social aspects are taken into consideration.

Note that the nature of game theory on social networks also differs
from \emph{cooperative games} (e.g.,~\cite{coop}) where each
coalition $C\subseteq 2^V$ of players $V$ has a certain
characteristic cost or payoff function $f: 2^V \rightarrow
\mathbb{R}$ describing the collective payoff the players can gain by
forming the coalition. In contrast to cooperative games, the
``coalitions'' are fixed, and a player participates in the
``coalitions'' of all its neighbors.

A preliminary version of this article appeared at ACM EC
2008~\cite{ec08}, and there have been several interesting results
related to our work since then. For example,~\cite{foes} studies
auctions with spite and altruism among bidders, and presents
explicit characterizations of Nash equilibria for first-price
auctions with random valuations and arbitrary spite/altruism
matrices, and for first and second price auctions with arbitrary
valuations and so-called regular social networks (players have same
out-degree). By rounding a natural linear program with
region-growing techniques, Chen et al.~\cite{ec10} present a better,
$O(\log z)$-approximation for the best vaccination strategy in the
original model of~\cite{virusgame}, where $z$ is the support size of
the outbreak distribution. Moreover, the effect of autonomy is
investigated: a benevolent authority may suggest which players
should be vaccinated, and the authors analyze the ``Price of Opting
Out'' under partially altruistic behavior; they show that with
positive altruism, Nash equilibria may not exist, but that the price
of opting out is bounded.

We extend the conference version of this article~\cite{ec08} in
several respects. The two most important additions concern
\emph{relative friendship} and \emph{convergence}. We study an
additional model where the relative importance of a neighbor
declines with the total number of friends and find that while
friendship is still always beneficial, the non-monotonicity result
no longer applies: unlike in the absolute friendship model, the
Windfall of Friendship can only increase with stronger social ties.
In addition, we initiate the study of convergence issues in the
social network. It turns out that it takes longer until an
equilibrium is reached compared to purely selfish environments and
hence constitutes a price of friendship. We present a potential
function argument to prove convergence in some simple cyclic
networks, and complement our study with simulations on Kleinberg
graphs. We believe that the existence of and convergence to social
equilibria are exciting questions for future research (see also the
related fields of \emph{player-specific utilities}~\cite{milchtaich}
and \emph{local search complexity}~\cite{pls}). Finally, there are
several minor changes, e.g., we improve the bound in
Theorem~\ref{thm:monotone} from $n>7$ to $n>3$.

\section{Model}\label{sec:model}

This section introduces our framework. In order to gain insights
into the Windfall of Friendship, we study a virus inoculation game
on a social graph. We present the model of this game and we show how
it can be extended to incorporate social aspects.

\subsection{Virus Inoculation Game}\label{ssec:virus}

The virus inoculation game was introduced by~\cite{virusgame}. We
are given an undirected network graph $G=(V,E)$  of $n=|V|$ players
(or nodes) $p_i\in V$, for $i=1,\dots,n$, who are connected by a set
of edges (or \emph{links}) $E$. Every player has to decide whether
it wants to \emph{inoculate} (e.g., purchase and install anti-virus
software) which costs $C$, or whether it prefers saving money  and
facing the risk of being infected. We assume that being infected
yields a damage cost of $L$ (e.g., a computer is out of work for $L$
days). In other words, an instance $I$ of a game consists of a graph
$G=(V,E)$, the inoculation cost $C$ and a damage cost $L$. We
introduce a variable $a_{i}$ for every player $p_i$ denoting $p_i$'s
chosen \emph{strategy}. Namely, $a_{i} = 1$ describes that player
$p_i$ is protected whereas for a player $p_j$ willing to take the
risk, $a_{j} = 0$. In the following, we will assume that $a_{j}\in
\{0,1\}$, that is, we do not allow players to \emph{mix} (i.e., use
probabilistic distributions over) their strategies. These choices
are summarized by the \textit{strategy profile}, the vector $\vec{a}
= (a_{1},\dots,a_{n})$. After the players have made their decisions,
a virus spreads in the network. The propagation model is as follows.
First, one player $p$ of the network is chosen uniformly at random
as a starting point. If this player is inoculated, there is no
damage and the process terminates. Otherwise, the virus infects $p$
and all unprotected neighbors of $p$. The virus now propagates
recursively to their unprotected neighbors. Hence, the more insecure
players are connected, the more likely they are to be infected. The
vulnerable region (set of players) in which an insecure player
$p_i$ lies is referred to as $p_i$'s \textit{attack component}. 

We only consider a limited region of the parameter space to avoid
trivial cases. If the cost $C$ is too large, no player will
inoculate, resulting in a totally insecure network and therefore all
players eventually will be infected. On the other hand, if $C<<L$,
the best strategy for all players is to inoculate. Thus, we will
assume that $C\leq L$ and $C> L/n$ in the following.

In our game, a player has the following expected cost:
\begin{defn}[Actual Individual Cost]\label{actualcost}$\\$
The \emph{actual individual cost} of a player $p_i$ is defined as
\begin{displaymath}
c_a(i,\vec{a}) = a_{i} \cdot C + (1-a_i) L \cdot \frac{k_i}{n}
\end{displaymath}
where $k_i$ denotes the size of $p_i$'s attack component. If $p_i$
is inoculated, $k_i$ stands for the size of the attack component
that would result if $p_i$ became insecure. In the following, let
$c_a^0(i,\vec{a})$ refer to the actual cost of an insecure and
$c_a^1(i,\vec{a})$ to the actual cost of a secure player $p_i$.
\end{defn}
The total \emph{social cost} of a game is defined as the sum of the
cost of all participants: $C_a(\vec{a}) = \sum_{p_i\in V}
c_a(i,\vec{a})$.

Classic game theory assumes that all players act selfishly, i.e.,
each player seeks to minimize its individual cost. In order to study
the impact of such selfish behavior, the solution concept of a
\emph{Nash equilibrium} (NE) is used. A Nash equilibrium is a
strategy profile where no selfish player can unilaterally reduce its
individual cost given the strategy choices of the other players. We
can think of Nash equilibria as the stable strategy profiles of
games with selfish players. We will only consider pure Nash
equilibria in this article, i.e., players cannot use random
distributions over their strategies but must decide whether they
want to inoculate or not.

In a pure Nash equilibrium, it must hold for each player $p_i$ that
given a strategy profile $\vec{a}$ $ \forall p_i\in V, ~\forall
a_i:~ c_a(i,\vec{a})\leq c_a(i,(a_{1},\ldots, 1-a_i, \ldots
,a_{n}))$, implying that player $p_i$ cannot decrease its cost by
choosing an alternative strategy $1-a_i$. In order to quantify the
performance loss due to selfishness, the (not necessarily unique)
Nash equilibria are compared to the optimum situation where all
players collaborate. To this end we consider the \emph{Price of
Anarchy} (PoA), i.e., the ratio of the social cost of the worst Nash
equilibrium divided by the optimal social cost for a problem
instance $I$. More formally,
$PoA(I)=\max_{NE}{C_{NE}(I)/C_{OPT}(I)}.$

\subsection{Social Networks}\label{ssec:families}

Our model for social networks is as follows. We define a
\textit{Friendship Factor} $F$ which captures the extent to which
players care about their \emph{friends}, i.e., about the players
\emph{adjacent} to them in the social network. More formally, $F$ is
the factor by which a player $p_i$ takes the individual cost of its
neighbors into account when deciding for a strategy.  $F$ can assume
any value between 0 and 1. $F=0$ implies that the players do not
consider their neighbors' cost at all, whereas $F=1$ implies that a
player values the well-being of its neighbors to the same extent as
its own. Let $\Gamma(p_i)$ denote the set of neighbors of a player
$p_i$. Moreover, let $\Gamma_{sec}(p_i)\subseteq \Gamma(p_i)$ be the
set of inoculated neighbors, and
$\Gamma_{\overline{sec}}(p_i)=\Gamma(p_i)\setminus
\Gamma_{sec}(p_i)$ the remaining insecure neighbors.

We distinguish between a player's \emph{actual cost} and a player's
\emph{perceived cost}. A player's actual individual cost is the
expected cost arising for each player defined in
Definition~\ref{actualcost} used to compute a game's social cost. In
our social network, the decisions of our players are steered by the
players' \emph{perceived cost}.
\begin{defn}[Perceived Individual Cost]\label{perceived cost}$\\$
The \emph{perceived individual cost} of a player $p_i$ is defined as
\begin{displaymath}
c_p(i,\vec{a}) = c_a(i,\vec{a}) + F \cdot \sum_{p_j \in \Gamma(p_i)}
c_a(j,\vec{a}).
\end{displaymath}
In the following, we write $c_p^0(i,\vec{a})$ to denote the
perceived cost of an insecure player $p_i$ and $c_p^1(i,\vec{a})$
for the perceived cost of an inoculated player.
\end{defn}

This definition entails a new notion of equilibrium. We define a
\emph{friendship Nash equilibrium} (FNE) as a strategy profile
$\vec{a}$  where no player can reduce its \emph{perceived} cost by
unilaterally changing its strategy given the strategies of the other
players. Formally, $\forall p_i\in V, \forall a_i:
c_p(i,\vec{a})\leq c_p(i,(a_{1},\ldots, 1-a_i, \ldots ,a_{n})).$
Given this equilibrium concept, we define the \emph{Windfall of
Friendship} $\Upsilon$.
\begin{defn}[Windfall of Friendship (WoF)]\label{def:WFDef}
The \textit{Windfall of Friendship} $\Upsilon(F,I)$ is the ratio of
the social cost of the worst Nash equilibrium for $I$ and the social
cost of the worst friendship Nash equilibrium for $I$:
$$\Upsilon(F,I)=\frac{\max_{NE}{C_{NE}(I)}}{\max_{FNE}{C_{FNE}(F,I)}}$$
\end{defn}

$\Upsilon(F,I) > 1$ implies the existence of a real windfall in the
system, whereas $\Upsilon(F,I) < 1$ denotes that the social cost can
become \emph{greater} in social graphs than in purely selfish
environments.

\section{General Analysis}\label{sec:general}

In this section we characterize friendship Nash equilibria and
derive general results on the Windfall of Friendship for the virus
propagation game in social networks. It has been
shown~\cite{virusgame} that in classic Nash equilibria ($F=0$), an
attack component can never consist of more than $Cn/L$ insecure
players. A similar characteristic also holds for friendship Nash
equilibria. As every player cares about its neighbors, the maximal
attack component size in which an insecure player $p_i$ still does
not inoculate depends on the number of $p_i$'s insecure neighbors
and the size of their attack components. Therefore, it differs from
player to player. We have the following helper lemma.
\begin{lem}\label{lemma:ac-size}
The player $p_i$ will inoculate if and only if the size of its
attack component is
\begin{displaymath}
k_i > \frac{Cn/L + F \cdot \sum_{p_j
\in{\Gamma_{\overline{sec}}(p_i)}}{k_j}}
{1+F|\Gamma_{\overline{sec}}(p_i)|},
\end{displaymath}
where the $k_j$s are the attack component sizes of $p_i$'s insecure
neighbors assuming $p_i$ is secure.
\end{lem}
\begin{proof}
Player $p_i$ will inoculate if and only if this choice lowers the
perceived cost. By Definition~\ref{perceived cost}, the perceived
individual cost of an inoculated player is
\begin{displaymath}
c_p^1(i,\vec{a}) = C + F \left( |\Gamma_{sec}(p_i)|{C} + \sum_{p_j
\in \Gamma_{\overline{sec}}(p_i)}{L \frac{k_j}{n}} \right)
\end{displaymath}
and for an insecure player we have
\begin{displaymath}
c_p^0(i,\vec{a}) = L \frac{k_i}{n} + F \left( |\Gamma_{sec}(p_i)|{C}
+ |\Gamma_{\overline{sec}}(p_i)|{L \frac{k_i}{n}} \right).
\end{displaymath}
For $p_i$ to prefer to inoculate it must hold that
\begin{eqnarray*}
c_p^0(i,\vec{a}) &>& c_p^1(i,\vec{a}) ~~\Leftrightarrow \\
L\frac{k_i}{n} + F \cdot |\Gamma_{\overline{sec}}(p_i)|{L
\frac{k_i}{n}} &>& C + F \cdot \sum_{p_j \in
\Gamma_{\overline{sec}}(p_i)}{L \frac{k_j}{n}} ~~\Leftrightarrow \\
L \frac{k_i}{n} (1 + F |\Gamma_{\overline{sec}}(p_i)|) &>& C +
\frac{FL}{n} \cdot \sum_{p_j \in \Gamma_{\overline{sec}}(p_i)}{k_j}
~~\Leftrightarrow \\
k_i (1 + F |\Gamma_{\overline{sec}}(p_i)|) &>& Cn/L + F \cdot
\sum_{p_j \in \Gamma_{\overline{sec}}(p_i)}{k_j}
~~\Leftrightarrow\\
k_i&>& \frac{Cn/L + F \cdot \sum_{p_j \in
\Gamma_{\overline{sec}}(p_i)}{k_j}}{1 + F
|\Gamma_{\overline{sec}}(p_i)|}.
\end{eqnarray*}
\end{proof}

A pivotal question is of course whether social networks where
players care about their friends yield better equilibria than
selfish environments. The following theorem answers this questions
affirmatively: the worst FNE costs never more than the worst NE.
\begin{thm}\label{FFneversmaller1}
For all instances of the virus inoculation game and $0\leq F\leq 1$,
it holds that
$$
1 \leq \Upsilon(F,I) \leq  \textit{PoA}(I)
$$
\end{thm}
\begin{proof}
The proof idea for $\Upsilon(F,I) \geq 1$ is the following: for an
instance $I$ we consider an arbitrary FNE with $F>0$. Given this
equilibrium, we show the existence of a NE with larger social cost
(according to~\cite{virusgame}, our best response strategy always
converges). Let $\vec{a}$ be any (e.g., the worst) FNE in the social
model. If $\vec{a}$ is also a NE in the same instance with $F=0$
then we are done. Otherwise there is at least one player $p_i$ that
prefers to change its strategy. Assume $p_i$ is insecure but favors
inoculation. Therefore $p_i$'s attack component has on the one hand
to be of size at least $k_i' > Cn/L$~\cite{virusgame} and on the
other hand of size at most $k_i'' = (Cn/L + F \cdot \sum_{p_j \in
\Gamma_{\overline{sec}}(p_i)}{k_j}) / (1+F
|\Gamma_{\overline{sec}}(p_i)|) \leq (Cn/L + F
|\Gamma_{\overline{sec}}(p_i)| (k_i'' - 1) ) / (1+F
|\Gamma_{\overline{sec}}(p_i)|) ~~\Leftrightarrow k_i'' \leq Cn/L -
F |\Gamma_{\overline{sec}}(p_i)|$ (cf Lemma~\ref{lemma:ac-size}).
This is impossible and yields a contradiction to the assumption that
in the selfish network, an additional player wants to inoculate.

It remains to study the case where $p_i$ is secure in the FNE but
prefers to be insecure in the NE. Observe that, since every player
has the same preference on the attack component's size when $F=0$, a
newly insecure player cannot trigger other players to inoculate.
Furthermore, only the players inside $p_i$'s attack component are
affected by this change. The total cost of this attack component
increases by at least
\begin{eqnarray*}
x &=& \underbrace{\frac{k_i}{n}L - C}_{p_i} + \underbrace{\sum_{p_j \in \Gamma_{\overline{sec}}(p_i)}\left(\frac{k_i}{n}L -\frac{k_j}{n}L\right)}_{\textit{$p_i$'s insecure neighbors}}\\
 &=& \frac{k_i}{n}L - C + \frac{L}{n}
(|\Gamma_{\overline{sec}}(p_i)|k_i - \sum_{p_j \in
\Gamma_{\overline{sec}}(p_i)}{k_j}).
\end{eqnarray*}

Applying Lemma~\ref{lemma:ac-size} guarantees that
$$\sum_{p_j \in
\Gamma_{\overline{sec}}(p_i)} k_j \leq
\frac{k_i(1+F|\Gamma_{\overline{sec}}(p_i)|)-C n/L}{F}.$$ This
results in
\begin{eqnarray*}
x &\geq& \frac{L}{n} \left(|\Gamma_{\overline{sec}}(p_i)|k_i -
\frac{ k_i(1+F|\Gamma_{\overline{sec}}(p_i)|)-C n/L}{F}\right)\\ &=&
\frac{k_i L}{n}(1-\frac{1}{F}) - C (1-\frac{1}{F})>0,
\end{eqnarray*}
since a player only gives up its protection if $C>\frac{k_i L}{n}$.
If more players are unhappy with their situation and become
vulnerable, the cost for the NE increases further. In conclusion,
there exists a NE for every FNE with $F\geq 0$ for the same instance
which is at least as expensive.

The upper bound for the WoF, i.e., $\textit{PoA($I$)} \geq
\Upsilon(F,I)$, follows directly from the definitions: while the PoA
is the ratio of the NE's social cost divided by the social optimum,
$\Upsilon(F,I)$ is the ratio between the cost of the NE and the FNE.
As the FNE's cost must be at least as large as the social optimum
cost the claim follows. \hfill\end{proof}

\begin{remark}
Note that Aspnes et al.~\cite{virusgame} proved that the Price of
Anarchy never exceeds the size of the network, i.e., $n \geq
\textit{PoA($I$)}$. Consequently, the Windfall of Friendship cannot
be larger than $n$ due to Theorem 4.2.
\end{remark}

The above result leads to the question of whether the Windfall of
Friendship grows monotonically with stronger social ties, i.e., with
larger friendship factors $F$. Intriguingly, this is not the case.
\begin{thm}\label{thm:monotone}
For all networks with more than three players, there exist game
instances where $\Upsilon(F,I)$ does not grow monotonically in $F$.
\end{thm}
\begin{proof}
We give a counter example for the star graph $S_n$ which has one
center player and $n-1$ leaf players. Consider two friendship
factors, $F_l$ and $F_s$ where $F_l>F_s$. We show that for the large
friendship factor $F_l$, there exists a FNE, $FNE_1$, where only the
center player and one leaf player remain insecure. For the same
setting but with a small friendship factor $F_s$, at least two leaf
players will remain insecure, which will trigger the center player
to inoculate, yielding a FNE, $FNE_2$, where only the center player
is secure.

Consider $FNE_1$ first. Let $c$ be the insecure center player, let
$l_{1}$ be the insecure leaf player, and let $l_{2}$ be a secure
leaf player. In order for $FNE_1$ to constitute a Nash equilibrium,
the following conditions must hold:
$$
\textit{player } c: ~~ \frac{2L}{n}+\frac{2F_lL}{n} < C +
\frac{F_lL}{n}
$$
$$
\textit{player } l_{1}: ~~ \frac{2L}{n}+\frac{2F_lL}{n} < C +
\frac{F_lL}{n}
$$
$$
\textit{player } l_{2}: ~~ C+\frac{2F_lL}{n} < \frac{3L}{n} +
\frac{3F_lL}{n}
$$

For $FNE_2$, let $c$ be the insecure center player, let $l_{1}$ be
one of the two insecure leaf players, and let $l_{2}$ be a secure
leaf player. In order for the leaf players to be happy with their
situation but for the center player to prefer to inoculate, it must
hold that:
$$
\textit{player } c: ~~ C+\frac{2F_sL}{n} < \frac{3L}{n} +
\frac{6F_sL}{n}
$$
$$
\textit{player } l_{1}: ~~ \frac{3L}{n}+\frac{3F_sL}{n} < C +
\frac{2F_sL}{n}
$$
$$
\textit{player } l_{2}: ~~ C+\frac{3F_sL}{n} < \frac{4L}{n} +
\frac{4F_sL}{n}
$$

Now choose $C:=5L/(2n)+F_lL/n$ (note that due to our assumption that
$n>3$, $C<L$). This yields the following conditions: $F_l>F_s+1/2$,
$F_l<F_s+3/2$, and $F_l<4F_s+1/2$. These conditions are easily
fulfilled, e.g., with $F_l=3/4$ and $F_s=1/8$. Observe that the
social cost of the first FNE (for $F_l$) is
$Cost(S_n,\vec{a}_{FNE1})=(n-2)C+4L/n$, whereas for the second FNE
(for $F_s$) $Cost(S_n,\vec{a}_{FNE2})=C+(n-1)L/n$. Thus,
$Cost(S_n,\vec{a}_{FNE1})-Cost(S_n,\vec{a}_{FNE2})=(n-3)C-(n-5)L/n>0$
as we have chosen $C>5L/(2n)$ and as, due to our assumption, $n>3$.
This concludes the proof. \hfill\end{proof}

Reasoning about best and worst Nash equilibria raises the question
of how difficult it is to compute such equlibria. We can generalize
the proof given in~\cite{virusgame} and show that computing the most
economical and the most expensive FNE is hard for any friendship
factor.
\begin{thm}\label{NP-completeness}
Computing the best and the worst pure FNE is $\mathcal{NP}$-complete
for any $F\in [0,1]$.
\end{thm}
\begin{proof}
We prove this theorem by a reduction from two $\mathcal{NP}$-hard
problems, \textsc{Vertex Cover}~\cite{vertexcover} and
\textsc{Independent Dominating Set}~\cite{garey79}. Concretely, for
the decision version of the problem, we show that answering the
question whether there exists a FNE costing less than $k$, or more
than $k$ respectively, is at least as hard as solving vertex cover
or independent dominating set. Note that verifying whether a
proposed solution is correct can be done in polynomial time, hence
the problems are indeed in $\mathcal{NP}$.

Fix some graph $G=(V,E)$ and set $C=1$ and $L=n/1.5$. We show that
the following two conditions are necessary and sufficient for a FNE:
(a) all neighbors of an insecure player are secure, and (b) every
inoculated player has at least one insecure neighbor. Due to our
assumption that $C>L/n$, condition (b) is satisfied in all FNE. To
see that condition (a) holds as well, assume the contrary, i.e., an
attack component of size at least two. An insecure player $p_i$ in
this attack component bears the cost $\frac{k_i}{n} L +
F(|\Gamma_{sec}(p_i)| C + |\Gamma_{\overline{sec}}(p_i)|
\frac{k_i}{n} L)$. Changing its strategy reduces its cost by at
least $\Delta_{i} = \frac{k_i}{n} L + F
|\Gamma_{\overline{sec}}(p_i)| \frac{k_i}{n} L - C - F
|\Gamma_{\overline{sec}}(p_i)| \frac{k_i-1}{n} L = \frac{k_i}{n} L +
F |\Gamma_{\overline{sec}}(p_i)| \frac{1}{n} L - C$. By our
assumption that $k_i \geq 2$, and hence
$|\Gamma_{\overline{sec}}(p_i)| \geq 1$, it holds that $\Delta_{i} >
0$, resulting in $p_i$ becoming secure. Hence, condition (a) holds
in any FNE as well. For the opposite direction assume that an
insecure player wants to change its strategy even though (a) and (b)
are true. This is impossible because in this case (b) would be
violated because this player does not have any insecure neighbors.
An inoculated player would destroy (a) by adopting another strategy.
Thus (a) and (b) are sufficient for a FNE.

We now argue that $G$ has a vertex cover of size $k$ if and only if
the virus game has a FNE with $k$ or fewer secure players, or
equivalently an equilibrium with social cost at most $Ck +
(n-k)L/n$, as each insecure player must be in a component of size 1
and contributes exactly $L/n$ expected cost. Given a minimal vertex
cover $V' \subseteq V$, observe that installing the software on all
players in $V'$ satisfies condition (a) because $V'$ is a vertex
cover and (b) because $V'$ is minimal. Conversely, if $V'$ is the
set of secure players in a FNE, then $V'$ is a vertex cover by
condition (a) which is minimal by condition (b).

For the worst FNE, we consider an instance of the independent
dominating set problem. Given an independent dominating set $V'$,
installing the software on all players except the players in $V'$
satisfies condition (a) because $V'$ is independent and (b) because
$V'$ is a dominating set. Conversely, the insecure players in any
FNE are independent by condition (a) and dominating by condition
(b). This shows that $G$ has an independent dominating set of size
at most $k$ if and only if it has a FNE with at least $n-k$ secure
players.\end{proof}

\section{Windfall for Special Graphs}\label{sec:cliquestar}

While the last section has presented general results on equilibria
in social networks and the Windfall of Friendship, we now present
upper and lower bounds on the Windfall of Friendship for concrete
topologies, namely the \emph{complete graph} $K_n$ and the
\emph{star graph} $S_n$.

\subsection{Complete Graphs}

In order to initiate the study of the Windfall of Friendship, we
consider a very simple topology, the complete graph $K_n$ where all
players are connected to each other. First consider the classic
setting where players do not care about their neighbors ($F=0$). We
have the following result:

\begin{lem}\label{cliqueNE}
In the graph $K_n$, there are two Nash equilibria with social cost
\begin{footnotesize}
\begin{eqnarray*}
\textit{~~NE$_1$: } Cost(K_n,\vec{a}_{\textit{NE1}}) &=& C(n-\lceil Cn/L \rceil + 1) + L/n(\lceil Cn/L \rceil - 1 )^2,\\ \textrm{and~~~~~~~~~~~~~~~~~~~~~~~~~~}\\
\textit{~~NE$_2$: }
 Cost(K_n,\vec{a}_{\textit{NE2}}) &=& C(n-\lfloor Cn/L \rfloor) + L/n(\lfloor Cn/L \rfloor)^2.
\end{eqnarray*}
\end{footnotesize}
If $\lceil Cn/L  \rceil- 1 = \lfloor Cn/L
\rfloor$, there is only one Nash equilibrium.
\end{lem}
\begin{proof}
Let $\vec{a}$ be a NE. Consider an inoculated player $p_i$ and an
insecure player $p_j$, and hence $c_a(i,\vec{a}) = C$ and
$c_a(j,\vec{a}) = L \frac{k_j}{n}$, where $k_j$ is the total number
of insecure players in $K_n$. In order for $p_i$ to remain
inoculated, it must hold that $C \leq (k_j+1)L/n$, so $k_j \geq
\lceil Cn/L-1 \rceil$; for $p_j$ to remain insecure, it holds that
$k_j L/n \leq C$, so $k_j \leq \lfloor Cn/L \rfloor$. As the total
social cost in $K_n$ is given by $Cost(K_n,\vec{a}) = (n-k_j)C +
k_j^2 L/n$, the claim follows. \hfill\end{proof}

Observe that the equilibrium size of the attack component is roughly
twice the size of the attack component of the social optimum, as
$Cost(K_n,\vec{a}) = (n-k_j)C + k_j^2 L/n$ is minimized for $k_j =
Cn/2L$.
\begin{lem}\label{cliqueopt}
In the social optimum for $K_n$, the size of the attack component is
either $\lfloor \frac{1}{2} Cn/L \rfloor$ or $\lceil \frac{1}{2}
Cn/L \rceil$, yielding a total social cost of
$$Cost(K_n,\vec{a}_{\textit{OPT}}) = (n-\lfloor \frac{1}{2} Cn/L
\rfloor)C + (\lfloor \frac{1}{2} Cn/L \rfloor)^2 \frac{L}{n}$$ or
$$Cost(K_n,\vec{a}_{\textit{OPT}}) = (n-\lceil \frac{1}{2} Cn/L
\rceil)C + (\lceil \frac{1}{2} Cn/L \rceil)^2 \frac{L}{n}.$$
\end{lem}

In order to compute the Windfall of Friendship, the friendship Nash
equilibria in social networks have to be identified.
\begin{lem}\label{cliqueFNE}
In $K_n$, there are two friendship Nash equilibria with social cost
\begin{footnotesize}
\begin{eqnarray*}
\textit{FNE$_1$: } Cost(K_n,\vec{a}_{\textit{FNE1}}) &=& C
\left(n-\left\lceil\frac{Cn/L - 1}{1 + F}\right\rceil\right) +
L/n\left(\left\lceil\frac{Cn/L - 1}{1 +
F}\right\rceil\right)^2, \\
\textrm{and~~~~~~~~~~~~~~~~~~~~~~~~~~~~~~~}\\
\textit{FNE$_2$: } Cost(K_n,\vec{a}_{\textit{FNE2}}) &=&
C\left(n-\left\lfloor\frac{Cn/L + F}{1 + F}\right\rfloor\right) +
L/n\left(\left\lfloor\frac{Cn/L + F}{1 + F}\right\rfloor\right)^2.
\end{eqnarray*}
\end{footnotesize}
If $\lceil(Cn/L - 1)/(1 + F)\rceil = \lfloor(Cn/L
+ F)/(1 + F)\rfloor$, there is only one FNE.
\end{lem}
\begin{proof}
According to Lemma~\ref{lemma:ac-size}, in a FNE, a player $p_i$
remains secure if otherwise the component had size at least $k_i =
k_j + 1 \geq (Cn/L + Fk_j^2)/(1+Fk_j)$ where $k_j$ is the number of
insecure players. This implies that $k_j \geq \lceil(Cn/L - 1)/(1 +
F)\rceil$. Dually, for an insecure player $p_j$ it holds that $k_j
\leq (Cn/L + F(k_j-1)^2)/(1+F(k_j-1))$ and therefore $k_j \leq
\lfloor(Cn/L + F)/(1 + F)\rfloor$. Given these bounds on the total
number of insecure players in a FNE, the social cost can be obtained
by substituting $k_j$ in $Cost(K_n,\vec{a}) = (n-k_j)C + k_j^2 L/n$.
As the difference between the upper and the lower bound for $k_j$ is
at most 1, there are at most two equilibria and the claim follows.
\hfill\end{proof}

Given the characteristics of the different equilibria, we have the
following theorem.
\begin{thm}\label{4/3}
In $K_n$, the Windfall of Friendship is at most $\Upsilon(F,I) =
4/3$ for an arbitrary network size. This is tight in the sense that
there are indeed instances where the worst FNE is a factor $4/3$
better than the worst NE.
\end{thm}
\begin{proof}
\emph{Upper Bound.} We first derive the upper bound on
$\Upsilon(F,I)$.
\begin{eqnarray*}
\Upsilon(F,I) &=&
\frac{Cost(K_n,\vec{a}_{\textit{NE}})}{Cost(K_n,\vec{a}_{\textit{FNE}})}\\
&\leq&
\frac{Cost(K_n,\vec{a}_{\textit{NE}})}{Cost(K_n,\vec{a}_{\textit{OPT}})}\\
&\leq& \frac{(n - \lceil Cn/L - 1 \rceil) C + (\lfloor Cn/L
\rfloor)^2 \frac{L}{n}}{(n - \frac{1}{2} Cn/L) C + (\frac{1}{2}
Cn/L)^2 \frac{L}{n}}
\end{eqnarray*}
as the optimal social cost (cf Lemma~\ref{cliqueopt}) is smaller or
equal to the social cost of any FNE. Simplifying this expression
yields
\begin{eqnarray*}
\Upsilon(F,I) &\leq& \frac{n(1-C/L)C + C^2 n/L}{n(1 -
\frac{1}{2}C/L)C + \frac{1}{4} C^2 n/L}=\frac{1}{1-\frac{1}{4}C/L}.
\end{eqnarray*}
This term is maximized for $L = C$, implying that $\Upsilon(F,I)
\leq 4/3$, for arbitrary $n$.

\emph{Lower Bound.} We now show that the ratio between the
equilibria cost reaches $4/3$.

There exists exactly one social optimum of cost $L n/2 + (n/2)^2
L/n=3nL/4$ for even $n$ and $C=L$ by Lemma~\ref{cliqueopt}. For
$F=1$ this is also the only friendship Nash equilibrium due to Lemma
\ref{cliqueFNE}. In the selfish game however the Nash equilibrium
has fewer inoculated players and is of cost $nL$ (see Lemma
\ref{cliqueNE}). Since these are the only Nash equilibria they
constitute the worst equilibria and the ratio is
$$
\Upsilon(F,I) =
\frac{Cost(K_n,\vec{a}_{\textit{NE}})}{Cost(K_n,\vec{a}_{\textit{FNE}})}=\frac{nL}{3/4nL}=4/3.
$$
\hfill\end{proof}

To conclude our analysis of $K_n$, observe that friendship Nash
equilibria always exist in complete graphs, and that in environments
where one player at a time is given the chance to change its
strategy in a best response manner quickly results in such an
equilibrium as all players have the same preferences.

\subsection{Star}

While the analysis of $K_n$ was simple, it turns out that already
slightly more sophisticated graphs are challenging. In the
following, we investigate the Windfall of Friendship in star graphs
$S_n$. Note that in $S_n$, the social welfare is maximized if the
center player inoculates and all other players do not. The total
inoculation cost then is $C$ and the attack components are all of
size $1$, yielding a total social cost of
$Cost(S_n,\vec{a}_{\textit{OPT}}) = C+(n-1)L/n$.

\begin{lem}\label{starOPT}
In the social optimum of the star graph $S_n$, only the center
player is inoculated. The social cost is
$$Cost(S_n,\vec{a}_{\textit{OPT}}) = C+(n-1)L/n.$$
\end{lem}

The situation where only the center player is inoculated also
constitutes a NE. However, there are more Nash equilibria.
\begin{lem}\label{starNE}
In the star graph $S_n$, there are at most three Nash equilibria
with social cost
\begin{footnotesize}
\begin{eqnarray*}
\textit{NE$_1$: }Cost(S_n,\vec{a}_{\textit{NE1}}) &=& C+(n-1)L/n,\\
\textit{NE$_2$: }Cost(S_n,\vec{a}_{\textit{NE2}}) &=& C(n-\lceil Cn/L \rceil +1 ) + L/n(\lceil Cn/L \rceil -1 )^2,\\
\textit{and~~~~~~~~~~~~~~~~~~~~~~~~~~~}\\
\textit{NE$_3$: }Cost(S_n,\vec{a}_{\textit{NE3}}) &=& C(n-\lfloor
Cn/L \rfloor) + L/n(\lfloor Cn/L \rfloor)^2.
\end{eqnarray*}
\end{footnotesize}

If $Cn/L\notin \mathbb{N}$, only two equilibria exist.
\end{lem}
\begin{proof}
If the center player is the only secure player, changing its
strategy costs $L$ but saves only $C$. When a leaf player becomes
secure, its cost changes from $L/n$ to $C$. These changes are
unprofitable, and the social cost of this NE is
$Cost(S_n,\vec{a}_{\textit{NE1}}) = C + (n-1)L/n$.

For the other Nash equilibria the center player is not inoculated.
Let the number of insecure leaf players be $n_0$. In order for a
secure player to remain secure, it must hold that $C \leq
(n_0+2)L/n$, and hence $n_0 \geq \lceil Cn/L-2 \rceil$. For an
insecure player to remain insecure, it must hold that $(1+n_0)L/n
\leq C$, thus $n_0 \leq \lfloor Cn/L-1 \rfloor$. Therefore, we can
conclude that there are at most two Nash equilibria, one with
$\lceil Cn/L-1 \rceil$ and one with $\lfloor Cn/L \rfloor$ many
insecure players. The total social cost follows by substituting
$n_0$ in the total social cost function. Finally, observe that if
$Cn/L \in \mathbb{N}$ and $Cn/L>3$, all three equilibria exist in
parallel; if $Cn/L \notin \mathbb{N}$, NE$_2$ and NE$_3$ become one
equilibrium.\hfill\end{proof}

Let us consider the social network scenario again.
\begin{lem}\label{starFNE}
In $S_n$, there are at most three friendship Nash equilibria with
social cost
\begin{footnotesize}
\begin{eqnarray*}
\textit{FNE$_1$: }Cost(S_n,\vec{a}_{\textit{FNE1}}) &=& C+(n-1)L/n,\\
\textit{FNE$_2$: }Cost(S_n,\vec{a}_{\textit{FNE2}}) &=& C(n-\lceil Cn/L -F \rceil +1 ) + L/n(\lceil Cn/L -F \rceil -1 )^2,\\
\textit{and~~~~~~~~~~~~~~~~~~~~~~~~~~~}\\
\textit{FNE$_3$: }Cost(S_n,\vec{a}_{\textit{FNE3}}) &=& C(n-\lfloor
Cn/L -F\rfloor) + L/n(\lfloor Cn/L-F \rfloor)^2.
\end{eqnarray*}
\end{footnotesize}
If $Cn/L-F\notin \mathbb{N}$, at most 2 friendship Nash equilibria
exist.
\end{lem}
\begin{proof}
First, observe that having only an inoculated center player
constitutes a FNE. In order for the center player to remain
inoculated, it must hold that $C + F(n-1)L\frac{1}{n} \leq nL/n +
F(n-1)L\frac{n}{n} =  L + F(n-1)L$. All leaf players remain insecure
as long as $L/n + FC \leq C + FC ~~\Leftrightarrow L/n \leq C$.
These conditions are always true, and we have
$Cost(S_n,\vec{a}_{\textit{FNE1}}) = C+(n-1)L/n$.
If the center player is not inoculated, we have $n_0$ insecure and
$n-n_0-1$ inoculated leaf players. In order for a secure leaf player
to remain secure, it is necessary that $C + F \frac{n_0+1}{n}L \leq
\frac{n_0+2}{n}L + F \frac{n_0+2}{n}L$, so $n_0 \geq \lceil Cn/L - F
\rceil -2$. For an insecure leaf player, it must hold that
$\frac{n_0+1}{n}L + F \frac{n_0+1}{n}L \leq C + F \frac{n_0}{n}L$,
so $n_0 \leq \lfloor Cn/L - F \rfloor -1$. The claim follows by
substitution. \hfill\end{proof}

Note that there are instances where FNE$_1$ is the only friendship
Nash equilibrium. We already made use of this phenomenon in Section
\ref{sec:general} to show that $\Upsilon(F,I)$ is not monotonically
increasing in $F$. The next lemma states under which circumstances
this is the case.

\begin{lem}\label{converge to best NE}\label{lemma:uniqueFNE}
In $S_n$, there is a unique FNE equivalent to the social optimum if
and only if
\begin{displaymath}
\lfloor Cn/L-F \rfloor - \lfloor \frac{1}{2F} (\sqrt{1-4F(1-Cn/L)}
-1 ) \rfloor -2 \geq 0
\end{displaymath}
\end{lem}
\begin{proof}
$S_n$ has only one FNE if and only if every (insecure) leaf player
is content with its chosen strategy but the insecure center player
would rather inoculate. In order for an insecure leaf player to
remain insecure we have $n_0 \leq \lfloor Cn/L - 1 - F \rfloor$ and
the insecure center player wants to inoculate if and only if
\begin{footnotesize}
\begin{eqnarray*} C +
F(n-n_0-1)C + Fn_0 \frac{1}{n}L < (n_0+1)\frac{L}{n} + F(n-n_0-1)C +
Fn_0 \frac{n_0+1}{n}L,\end{eqnarray*} \end{footnotesize} which is
equivalent to $ Fn_0^2 + n_0 + 1 - Cn/L > 0.$ This implies that $n_0
\geq \lfloor \frac{1}{2F} (\sqrt{1-4F(1-Cn/L)} -1 ) +1 \rfloor.$
Therefore there is only one FNE if and only if there exists an
integer $n_0$ such that $\lfloor \frac{1}{2F} (\sqrt{1-4F(1-Cn/L)}
-1 ) +1 \rfloor \leq n_0 \leq \lfloor Cn/L - 1 - F \rfloor$.
\hfill\end{proof}

Given the characterization of the various equilibria, the Windfall
of Friendship can be computed.
\begin{thm}\label{thm:starFF}
If $\lfloor \frac{1}{2F} (\sqrt{1-4F(1-Cn/L)} -1 ) \rfloor
+2-\lfloor Cn/L-F \rfloor \leq 0$,
 the Windfall of Friendship is $$\Upsilon(F,I)\geq \frac{(n-2)C+L/n}{C+(n-1)L/n}, \textit{ ~~else~~ }\Upsilon(F,I)\leq \frac{n+1}{n-3}.$$
\end{thm}
\begin{proof}
According to Lemma~\ref{lemma:uniqueFNE}, the friendship Nash
equilibrium is unique and hence equivalent to the social optimum if
$$\lfloor Cn/L-F \rfloor - \lfloor \frac{1}{2F} (\sqrt{1-4F(1-Cn/L)}
-1 ) \rfloor -2 \geq 0.$$  On the other hand, observe that there
always exist sub-optimal Nash equilibria where the center player is
not inoculated. Hence, we have
\begin{eqnarray*}
\Upsilon(F,I) &=& \frac{Cost(S_n,\vec{a}_{\textit{NE}})}{Cost(S_n,\vec{a}_{\textit{FNE}})}= \frac{Cost(S_n,\vec{a}_{\textit{NE}})}{Cost(S_n,\vec{a}_{\textit{OPT}})}\\
&\geq& \frac{(n-\lfloor Cn/L-1 \rfloor)C + (\lceil Cn/L \rceil - 1)^2 L/n}{C+(n-1)L/n}\\
&\geq& \frac{C(n-2)+L/n}{C+(n-1)L/n}.
\end{eqnarray*}

Otherwise, i.e., if there exist friendship Nash equilibria with an
insecure center player, an upper bound for the WoF can be computed
\begin{eqnarray*}
\Upsilon(F,I) &=& \frac{Cost(S_n,\vec{a}_{\textit{NE}})}{Cost(S_n,\vec{a}_{\textit{FNE}})}\\
&\leq& \frac{(n-\lceil Cn/L-1 \rceil)C + (\lfloor Cn/L \rfloor)^2 L/n}{(n-\lfloor Cn/L - F \rfloor)C + (\lceil Cn/L - 1 - F \rceil)^2 L/n}\\
&\leq& \frac{(n+1)C}{nC+FC-2C(1+F)+(1+F)^2 L/n}\\
&<& \frac{(n+1)C}{C(n+F-2(1+F))}~~~<~~~ \frac{n+1}{n-3}.
\end{eqnarray*}
\hfill\end{proof}

Theorem~\ref{thm:starFF} reveals that caring about the cost incurred
by friends is particularly helpful to reach more desirable
equilibria. In large star networks, the social welfare can be much
higher than in Nash equilibria: in particular, the Windfall of
Friendship can increase linearly in $n$, and hence indeed be
asymptotically as large as the Price of Anarchy. However, if
$\lfloor Cn/L-F \rfloor - \lfloor \frac{1}{2F} (\sqrt{1-4F(1-Cn/L)}
-1 )  \rfloor-2 \geq 0$ does not hold, social networks are not much
better than purely selfish systems: the maximal gain is constant.

Finally observe that in stars friendship Nash equilibria always
exist and can be computed efficiently (in linear time) by any best
response strategy.

\subsection{Discussion}

This section has focused on a small set of very simple topologies
only and we regard the derived results as a first step towards more
complex graph classes such as Kleinberg graphs featuring the
small-world property. Interestingly, however, our findings have
implications for general topologies. We could show that even in
simple graphs such as the star graph, the Windfall of Friendship can
assume all possible values, from constant ratios up to ratios linear
in $n$. This is asymptotically maximal for general graphs as well
since the Price of Anarchy is bounded by $n$~\cite{virusgame}.

\section{On Relative Equilibria}\label{sec:relative}

In the model we have studied so far, the actual cost of each
friend---weighted by a factor $F$---is added to a player's perceived
cost. This describes a situation where friends are taken into
account individually and independently of each other. However, one
could imagine scenarios where the relative importance of a friend
decreases with the total number of friends, that is, a player with
many friends may care less about the welfare of a specific friend
compared to a player who only has one or two friends. This motivates
an alternative approach to describe perceived costs:
\begin{defn}[Relative Perceived Cost]\label{rel_perceived cost}$\\$
The \emph{relative perceived individual cost} of a player $p_i$ is
defined as
\begin{displaymath}
c_r(i,\vec{a}) = c_a(i,\vec{a}) + F \cdot \frac{\sum_{p_j \in
\Gamma(p_i)} c_a(j,\vec{a})}{|\Gamma(p_i)|}.
\end{displaymath}
In the following, we write $c_r^0(i,\vec{a})$ to denote the relative
perceived cost of an insecure player $p_i$ and $c_r^1(i,\vec{a})$
for the relative perceived cost of an inoculated player.
\end{defn}
We will refer to an FNE equilibrium with respect to relative
perceived costs by \emph{rFNE}.

It turns out that while relative equilibria have similar properties
as regular friendship equilibria and most of our techniques are
still applicable, there are some crucial differences. Let us first
consider the size of the attack components under rFNE.
\begin{lem}\label{lemma:ac-size-rel}
The player $p_i$ will inoculate if and only if the size of its
attack component is
\begin{displaymath}
k_i > \frac{|\Gamma(p_i)|\cdot Cn/L + F \cdot \sum_{p_j \in
\Gamma_{\overline{sec}}(p_i)}{k_j}}{|\Gamma(p_i)| + F
|\Gamma_{\overline{sec}}(p_i)|},
\end{displaymath}
where the $k_j$s are the attack component sizes of $p_i$'s insecure
neighbors assuming $p_i$ is secure.
\end{lem}
\begin{proof}
Player $p_i$ will inoculate if and only if this choice lowers the
relative perceived individual cost. By Definition~\ref{rel_perceived
cost}, the relative perceived individual costs of an inoculated
player are
\begin{displaymath}
c_r^1(i,\vec{a}) = C + F/|\Gamma(p_i)| \cdot \left(
|\Gamma_{sec}(p_i)|{C} + \sum_{p_j \in
\Gamma_{\overline{sec}}(p_i)}{L \frac{k_j}{n}} \right)
\end{displaymath}
and for an insecure player we have
\begin{displaymath}
c_p^0(i,\vec{a}) = L \frac{k_i}{n} + F/|\Gamma(p_i)|\cdot \left(
|\Gamma_{sec}(p_i)|{C} + |\Gamma_{\overline{sec}}(p_i)|{L
\frac{k_i}{n}} \right).
\end{displaymath}
Thus, for $p_i$ to prefer to inoculate it must hold that
\begin{eqnarray*}
c_p^0(i,\vec{a}) &>& c_p^1(i,\vec{a}) ~~\Leftrightarrow \\
k_i&>& \frac{Cn/L + F/|\Gamma(p_i)| \cdot \sum_{p_j \in
\Gamma_{\overline{sec}}(p_i)}{k_j}}{1 + F/|\Gamma(p_i)|\cdot
|\Gamma_{\overline{sec}}(p_i)|}.
\end{eqnarray*}
\end{proof}

Not surprisingly, we can show that friendship is always beneficial
also with respect to relative perceived costs.
\begin{thm}\label{FFneversmaller1relative}
For all instances of the virus inoculation game and $0\leq F\leq 1$,
it holds that
$$
1 \leq \Upsilon(F,I) \leq  \textit{PoA}(I)
$$
also in the relative cost model.
\end{thm}
\begin{proof}
Again, the upper bound for the WoF, i.e., $\textit{PoA($I$)} \geq
\Upsilon(F,I)$, follows directly from the definitions (see also
proof of Lemma~\ref{FFneversmaller1}). For $\Upsilon(F,I) \geq 1$ we
start from a rFNE $\vec{a}$ (defined with relative costs) with $F>0$
and show that a best response execution yields a Nash equilibrium
$\vec{a}'$ with cost $C_a(\vec{a})\leq C_a(\vec{a}')$. If $\vec{a}$
is also a NE in the same instance with $F=0$ then we are done.
Otherwise there is at least one player $p_i$ that prefers to change
its strategy. If $p_i$ is insecure but favors inoculation, $p_i$'s
attack component has on the one hand to be of size at least $k_i' >
Cn/L$~\cite{virusgame} (otherwise there is not reason for $p_i$ to
become secure) and on the other hand of size at most $k_i'' =
(|\Gamma(p_i)|\cdot Cn/L + F \cdot \sum_{p_j \in
\Gamma_{\overline{sec}}(p_i)}{k_j}) / (|\Gamma(p_i)|+F\cdot
|\Gamma_{\overline{sec}}(p_i)|) \leq (|\Gamma(p_i)|\cdot Cn/L +
F\cdot |\Gamma_{\overline{sec}}(p_i)| (k_i'' - 1) ) /
(|\Gamma(p_i)|+F\cdot|\Gamma_{\overline{sec}}(p_i)|)$ so $k_i'' \leq
|\Gamma(p_i)|\cdot Cn/L - F|\Gamma_{\overline{sec}}(p_i)|$ (cf
Lemma~\ref{lemma:ac-size-rel}), yielding a contradiction. What if
$p_i$ is secure in the rFNE but prefers to be insecure in the NE?
Since every player has the same preference on the attack component's
size when $F=0$, a newly insecure player cannot trigger other
players to inoculate. Furthermore, only the players inside $p_i$'s
attack component are affected by this change. The total cost of this
attack component increases by at least (see also the proof of
Lemma~\ref{FFneversmaller1})
\begin{eqnarray*}
x = \frac{k_i}{n}L - C + \frac{L}{n}
(|\Gamma_{\overline{sec}}(p_i)|k_i - \sum_{p_j \in
\Gamma_{\overline{sec}}(p_i)}{k_j}).
\end{eqnarray*}
Applying Lemma~\ref{lemma:ac-size-rel} guarantees that
$$\sum_{p_j \in
\Gamma_{\overline{sec}}(p_i)} k_j \leq
\frac{k_i(1+F/|\Gamma(p_i)|\cdot|\Gamma_{\overline{sec}}(p_i)|)-C
n/L}{F/|\Gamma(p_i)|}.$$ This results in
\begin{eqnarray*}
x &\geq& \frac{L}{n} \left(|\Gamma_{\overline{sec}}(p_i)|k_i -
\frac{ k_i(1+F/|\Gamma(p_i)|\cdot|\Gamma_{\overline{sec}}(p_i)|)-C
n/L}{F/|\Gamma(p_i)|\cdot}\right)\\ &=& \frac{k_i
L}{n}\left(1-\frac{1}{F/|\Gamma(p_i)|\cdot}\right) - C
\left(1-\frac{1}{F/|\Gamma(p_i)|\cdot}\right)>0,
\end{eqnarray*}
since a player only gives up its protection if $C>\frac{k_i L}{n}$.
If more players are unhappy with their situation and become
vulnerable, the cost for the NE increases further. In conclusion,
there exists a NE for every FNE with $F\geq 0$ for the same instance
which is at least as expensive. \hfill\end{proof}

Interestingly, however, the phenomenon of a non-monotonic welfare
increase with larger $F$ does no longer hold in the star graph
$S_n$. To see this, note that there are only at most two distinct
rFNE in $S_n$ (apart from the trivial situations where all players
are either insecure or secure): the ``good equilibrium'' where the
center player is secure and all the leave players insecure, and the
``bad equilibrium'' where the center is insecure and a fraction of
the leaves secure. The following theorem shows that the example of
Theorem~\ref{thm:monotone} for FNE is no longer true for rFNE.
\begin{thm}\label{thm:rel-monotone}
The Windfall of Friendship is monotonic in $F$ for $S_n$ under the
relative cost model.
\end{thm}
\begin{proof}
Consider a friendship factor $F$. Clearly, the equilibrium where
only the center player is secure always exists (w.l.o.g., we focus
on ``reasonable values'' $C$ and $L$). When is there an equilibrium
where the center is insecure? Consider such an equilibrium where $x$
leave players are insecure. In order for this to constitute an
equilibrium, it must hold for the center player that:
\begin{footnotesize}
$$
\frac{(x+1)L}{n}+\frac{F}{n-1}\cdot \frac{(x+1)L}{n}+\frac{F\cdot
C\cdot (n-x-1)}{n-1}< C + \frac{F}{n-1}\cdot \frac{x\cdot
L}{n}+\frac{F\cdot C\cdot (n-x-1)}{n-1}
$$
$$
\Leftrightarrow \frac{(x+1)L}{n}+\frac{F}{n-1}\cdot \frac{L}{n}<C
$$
\end{footnotesize}
On the other hand, for an insecure leaf player we have:
$$
\frac{(x+1)L}{n}+\frac{FL(x+1)}{n}< C + \frac{FLx}{n}
$$
$$
\Leftrightarrow \frac{(x+1)L}{n}+\frac{FL}{n}<C
$$
Unlike in the FNE scenario, the center player is less likely to
inoculate, i.e., leaf players inoculate first. Thus, a larger $F$
can only render the existence of such an equilibrium more unlikely.
\hfill\end{proof}

Finally, note that the hardness result of
Theorem~\ref{NP-completeness} is also applicable to relative FNEs.
\begin{thm}\label{rNP-completeness}
Computing the best and the worst pure rFNE is $\mathcal{NP}$-complete
for any $F\in [0,1]$.
\end{thm}
\begin{proof} (\emph{Sketch})
Again, deciding the existence of a rFNE with cost less than $k$ or
more than $k$ is at least as hard as solving the \emph{vertex cover}
or \emph{independent dominating set} problem, respectively. Note
that verifying whether a proposed solution is correct can be done in
polynomial time, hence the problems are indeed in $\mathcal{NP}$.
The proof is similar to Theorem~\ref{NP-completeness}, and we only
point out the difference for condition (a): an insecure player $p_i$
in the attack component bears the cost $k_i/n\cdot L +
F|\Gamma_{sec}(p_i)| C + |\Gamma_{\overline{sec}}(p_i)| \cdot
(k_iL/n)/|\Gamma(p_i)|$, and changing its strategy reduces the cost
by at least $\Delta_{i} = k_iL/n+ F |\Gamma_{\overline{sec}}(p_i)|
k_i L/ (|\Gamma(p_i)|  n)  - C - F |\Gamma_{\overline{sec}}(p_i)|
(k_i-1) L/(|\Gamma(p_i)| n) = k_iL/n - C + F L
|\Gamma_{\overline{sec}}(p_i)|/(|\Gamma(p_i)| n)$. By our assumption
that $k_i \geq 2$, and hence $|\Gamma_{\overline{sec}}(p_i)| \geq
1$, it holds that $\Delta_{i} > 0$, resulting in $p_i$ becoming
secure.
\end{proof}

\section{Convergence}\label{sec:convergence}

According to Lemma~\ref{FFneversmaller1} and
Lemma~\ref{FFneversmaller1relative}, the social context can only
improve the overall welfare of the players, both in the absolute and
the relative friendship model. However, there are implications
beyond the players' welfare in the equilibria: in social networks,
the dynamics of how the equilibria are reached is different.

In~\cite{virusgame}, Aspnes et al.~have shown that best-response
behavior quickly leads to some pure Nash equilibrium, from any
initial situation. Their potential function argument however relies
on a ``symmetry'' of the players in the sense insecure players in
the same attack component have the same cost. This no longer holds
in the social context where different players take into account
their neighborhood: a player with four insecure neighbors is more
likely to inoculate than a player with just one, secure neighbor.
Thus, the distinction between ``big'' and ``small'' components used
in~\cite{virusgame} cannot be applied, as different players require
a different threshold.

Nevertheless, convergence can be shown in certain scenarios. For
example, the hardness proofs of Lemmas~\ref{NP-completeness}
and~\ref{rNP-completeness} imply that equilibria always exist in the
corresponding areas of the parameter space, and it is easy to see
that the equilibria are also reached by best-response sequences.
Similarly, in the star and complete networks, best-response
sequences converge in linear time. Linear convergence time also
happens in more complex, cyclic graphs. For example, consider the
cycle graph $C_n$ where each player is connected to one left and one
right neighbor in a circular fashion. To prove best response
convergence from arbitrary initial states, we distinguish between an
initial phase where certain structural invariants are established,
and a second phase where a potential function argument can be
applied with respect to the view of only one type of players. Each
event when one player is given the chance to perform a best response
is called a \emph{round}.

\begin{thm}\label{thm:cycconv}
From any initial state and in the cycle graph $C_n$, a best response
round-robin sequence results in an equilibrium after $O(n)$ changes,
both in case of absolute and relative friendship equilibria.
\end{thm}
\begin{proof}
After two round-robin phases where each player is given the chance
to make a best response twice (at most $2n$ changes or rounds), it
holds that an insecure player $p_1$ which is adjacent to a secure
player $p_2$ cannot become secure: since $p_1$ preferred to be
insecure at some time $t$, the only reason to become secure again is
the event that a player $p_3$ becomes insecure in $p_1$'s attack
component at time $t'>t$; however, since $p_1$ has a secure neighbor
$p_2$ and hence $p_3$ can only have more insecure neighbors than
$p_1$, $p_3$ cannot prefer a larger attack component than $p_1$,
which yields a contradiction to the assumption that $p_1$ becomes
secure while its neighbor $p_2$ is still secure. Moreover, by the
same arguments, there cannot be three consecutive secure players.

Therefore, in the best response rounds after the two initial phases,
there are the following cases. Case~(A): a secure player having two
insecure neighbors becomes insecure; Case~(B): a secure player with
one secure neighbor becomes insecure; and Case~(C): an insecure
player with two insecure neighbors becomes secure.

In order to prove convergence, the following potential function
$\Phi$ is used:
$$
\Phi(\vec{a}) = \sum_{A\in \mathcal{S}_{big}(\vec{a})} |A| -
\sum_{A\in \mathcal{S}_{small}(\vec{a})} |A|
$$
where the attack components $A$ in $\mathcal{S}_{big}$ contain more
than $t=nC/(FL)-L/F+1$ players and the attack components $A$ in
$\mathcal{S}_{small}$ contain at most $t$ players in case of
absolute friendship equilibria; for relative friendship equilibria
we use $t=2Cn/(FL)-2L/F+1$. In other words, the threshold $t$ to
distinguish between small and big components is chosen with respect
to players having \emph{two insecure neighbors}; in case of absolute
FNEs:
$$
C+F\cdot\frac{L\cdot(t-1)}{n}=\frac{t\cdot L}{n}+2F\cdot\frac{L\cdot
t}{n} \Leftrightarrow \frac{Cn}{FL}-\frac{L}{F}+1=t
$$
and in case of relative FNEs:
$$
C+F/2\cdot\frac{L\cdot(t-1)}{n}=\frac{t\cdot
L}{n}+F\cdot\frac{L\cdot t}{n} \Leftrightarrow
\frac{2Cn}{FL}-\frac{2L}{F}+1=t
$$

Note that it holds that $-n\leq\Phi(\vec{a})\leq n,
\forall~\vec{a}$. We now show that Case~(A) and~(C) reduce
$\Phi(\vec{a})$ by at least one unit in each best response.
Moreover, Case~(B) can increase the potential by at most one.
However, since we have shown that Case~(B) incurs less than $n$
times, the claim follows by an amortization argument.
\emph{Case~(A):} In this case, a new insecure player $p_1$ is added
to an attack component in $\mathcal{S}_{small}$. \emph{Case~(B):} A
new insecure player $p_1$ is added to an attack component in
$\mathcal{S}_{small}$ or to an attack component in
$\mathcal{S}_{big}$ (since $p_1$ is ``on the edge'' of the attack
component, it prefers a larger attack component). \emph{Case~(C):}
An insecure player is removed from an attack component in
$\mathcal{S}_{big}$.
\end{proof}
The proof of Theorem~\ref{thm:cycconv} can be adapted to show linear
convergence in general 2-degree networks where players have degree
at most two. In order to gain deeper insights into the convergence
behavior, we conducted several experiments.

\section{Simulations}\label{sec:simulations}

This section briefly reports on the simulations conducted on
Kleinberg graphs (using clustering exponent $\alpha=2$). Although
the existence of equilibria and the best-response convergence time
complexity for general graphs remain an open question, during the
thousands of experiments, we did not encounter a single instance
which did not converge. Moreover, our experiments indicate that the
initial configuration (i.e., the set of secure and insecure players)
as well as the relationship of $L$ to $C$ typically has a negligible
effect on the convergence time, and hence, unless stated otherwise,
the following experiments assume an initially completely insecure
network and $C=1$ and $L=4$. All experiments are repeated 100 times
over different Kleinberg graphs.

All our experiments showed a positive Windfall of Friendship that
increases monotonically in $F$, both for the relative and the
absolute friendship model. Figure~\ref{fig:socialcost} shows a
typical result. Maybe surprisingly, it turns out that the windfall
of friendship is often not due to a higher fraction of secure
players, but rather the fact that the secure players are located at
strategically more beneficial locations (see also
Figure~\ref{fig:numsec}). We can conclude that there is a windfall
of friendship not only for the worst but also for ``average
equilibria''.
\begin{figure*} [ht]
\begin{center}
\includegraphics[width=0.595\textwidth]{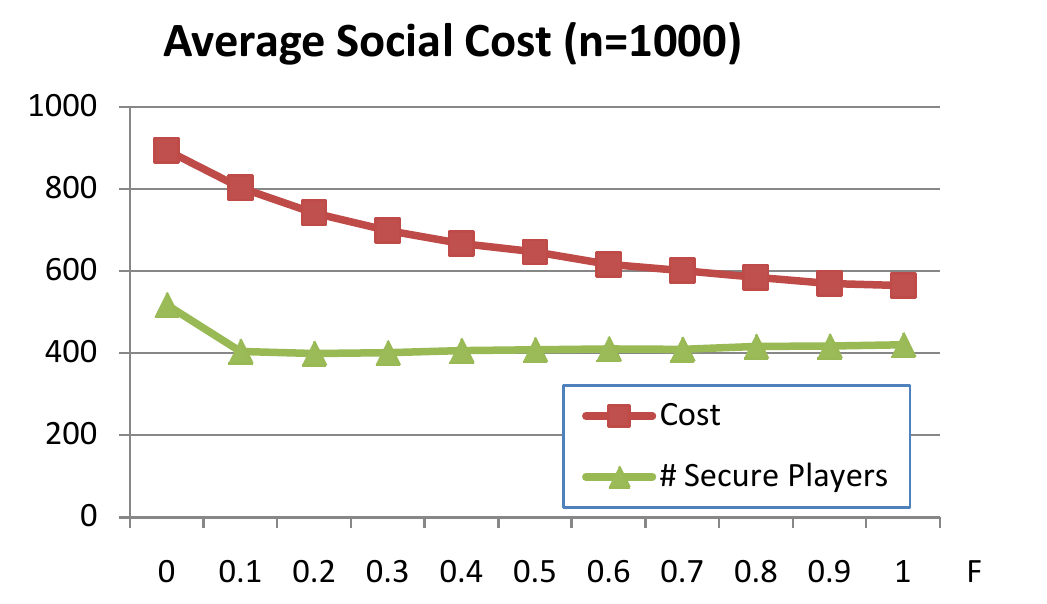}\\
\caption{Average social cost and average number of secure players as
a function of $F$, in the FNE resulting from round-robin best
response sequences starting from an initially completely insecure
graph.}\label{fig:socialcost}
\end{center}
\end{figure*}
\begin{figure*} [ht]
\begin{center}
\includegraphics[width=0.895\textwidth]{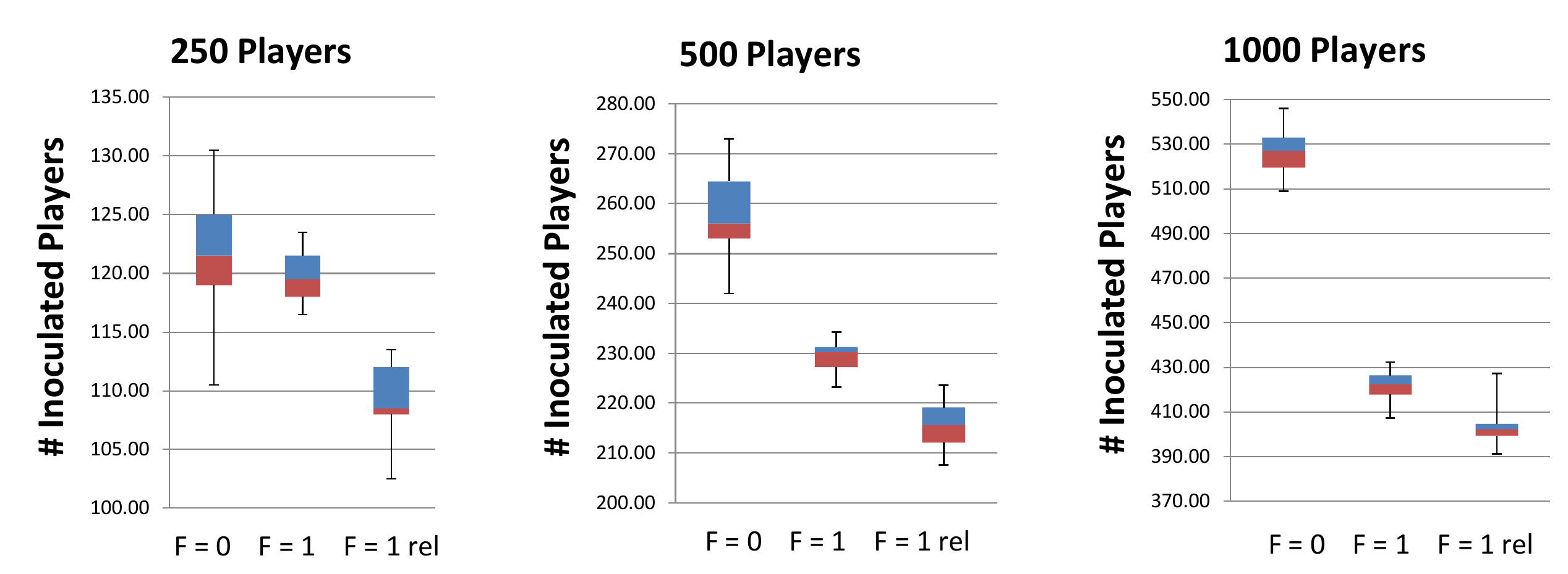}\\
\caption{Number of secure players in different models using $L=16$:
friendship often does not increase the number but yields better
locations of the secure players.}\label{fig:numsec}
\end{center}
\end{figure*}

The box plots in Figure~\ref{fig:boxplotcost} give a more detailed
picture of the cost for $F\in\{0,1\}$. The overall cost of pure NE
is typically higher than the cost of rFNE which is in turn higher
than the cost of FNE.
\begin{figure*} [ht]
\begin{center}
\includegraphics[width=0.895\textwidth]{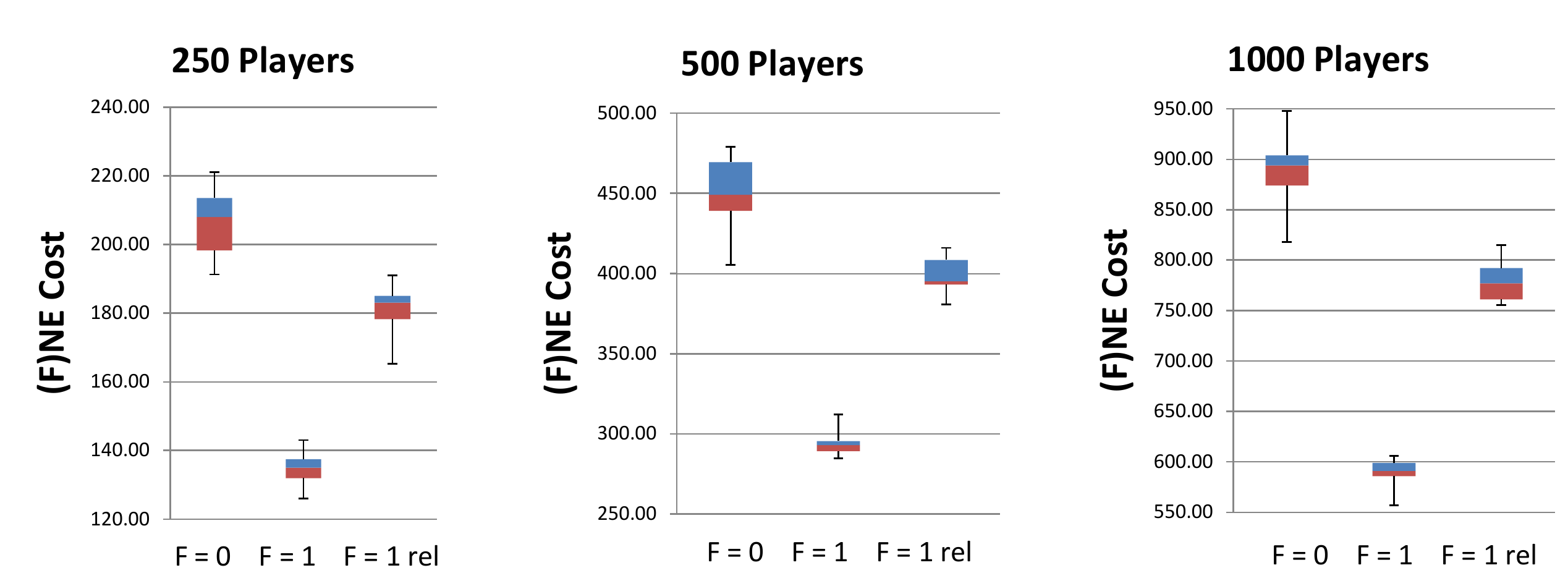}\\
\caption{Box plots of social cost in different scenarios. The
considered equilibria resulted from round-robin best response
sequences starting from an initially completely insecure
graph.}\label{fig:boxplotcost}
\end{center}
\end{figure*}

Besides social cost, we are mainly interested in convergence times.
We find that while the convergence time typically increases already
for a small $F>0$, the magnitude of $F$ plays a minor role.
Figure~\ref{fig:convf} shows the typical convergence times as a
function of $F$. Notice that the convergence time more than doubles
when changing from the selfish to the social model but is roughly
constant for all values of $F$.
\begin{figure*} [ht]
\begin{center}
\includegraphics[width=0.595\textwidth]{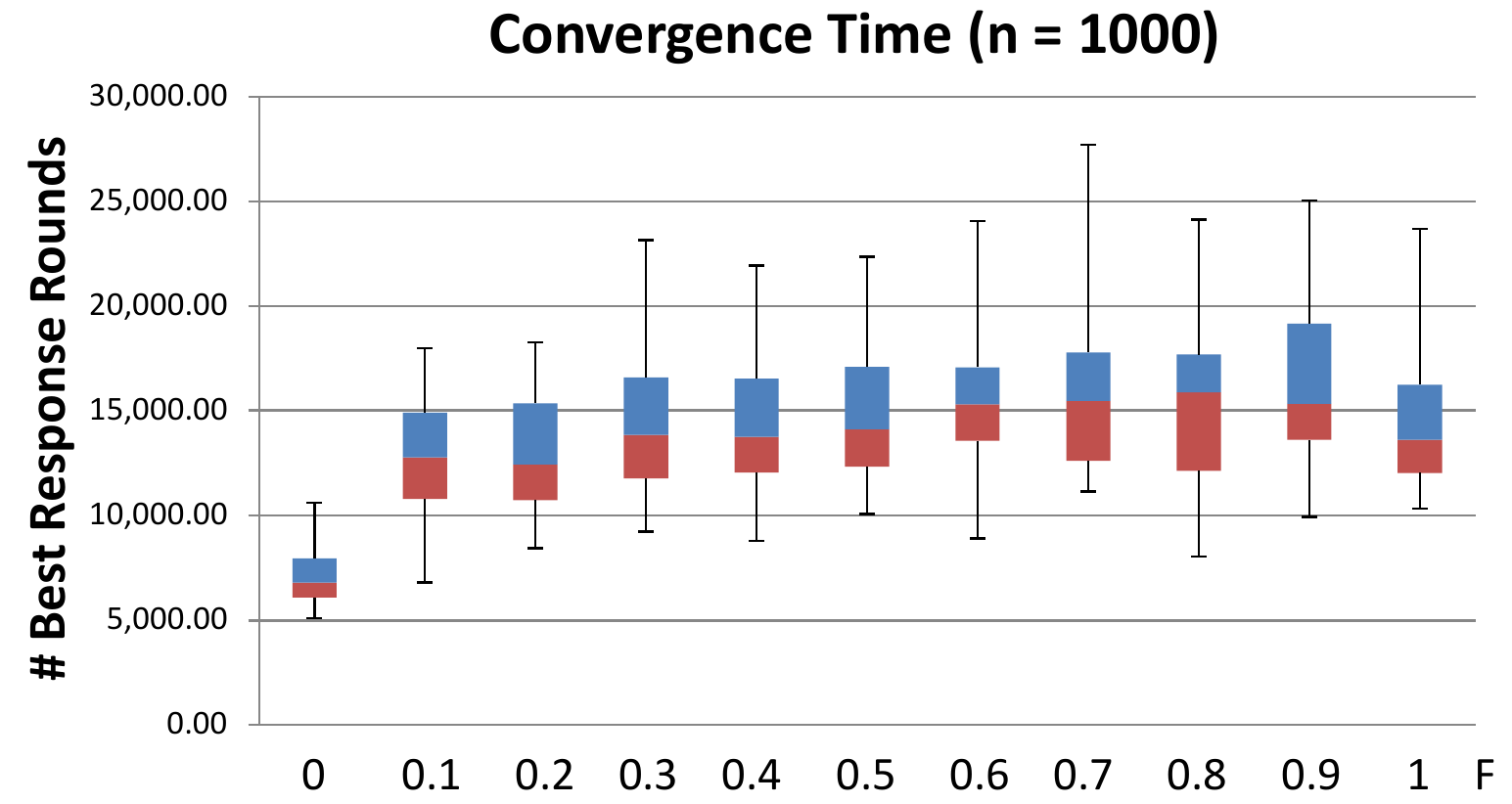}\\
\caption{Box plot of number of best response rounds until
convergence to FNE, starting from an initially completely insecure
graph.}\label{fig:convf}
\end{center}
\end{figure*}



\section{Conclusion}\label{sec:conclusion}

This article presented a framework to study and quantify the effects
of game-theoretic behavior in social networks. This framework allows
us to formally describe and understand phenomena which are often
well-known on an anecdotal level. For instance, we find that the
Windfall of Friendship is always positive, and that players embedded
in a social context may be subject to longer convergence times.
Moreover, interestingly, we find that the Windfall of Friendship
does not always increase monotonically with stronger social ties.

We believe that our work opens interesting directions for future
research. We have focused on a virus inoculation game, and
additional insights must  be gained by studying alternative and more
general games such as potential games, or games that do and do not
exhibit a Braess paradox. Also the implications on the games'
dynamics need to be investigated in more detail, and it will be
interesting to take into consideration behavioral models beyond
equilibria (e.g.,~\cite{guest2}). Finally, it may be interesting to
study scenarios where players care not only about their friends but
also, to a smaller extent, about friends of friends.

What about practical implications? One intuitive takeaway of our
work is that in case of large benefits of social behavior, it may
make sense to design distributed systems where neighboring players
have good relationships. However, if the resulting convergence times
are large and the price of the dynamics higher than the possible
gains, such connections should be discouraged. Our game-theoretic
tools can be used to compute these benefits and convergence times,
and may hence be helpful during the design phase of such a system.

\section*{Acknowledgments}

We would like to thank Yishay Mansour and Boaz Patt-Shamir from Tel
Aviv University and Martina H\"ullmann and Burkhard Monien from
Paderborn University for interesting discussions on relative
friendship equilibria and aspects of convergence.

{ \balance

}


\end{document}